\documentclass[letterpaper, 10 pt, conference]{ieeeconf}  

\IEEEoverridecommandlockouts                              
\usepackage{graphicx}      
\usepackage{amsmath,amssymb}
\usepackage{setspace}
\usepackage{enumerate}
\usepackage{comment}
\usepackage{curve2e}
\usepackage{multirow}
\usepackage{subfigure}
\usepackage{graphicx}
\usepackage{booktabs} 

\usepackage[T1]{fontenc}
\usepackage[utf8]{inputenc}
\usepackage{multirow}
\usepackage{mathrsfs}
\usepackage{framed}
\usepackage{url}
\newfam\bbfam
\font\tenbb=msbm10
\textfont\bbfam=\tenbb

\newtheorem{thm}{Theorem}[section]

\newtheorem{prop}[thm]{Proposition}

\usepackage{algorithm}
\usepackage{algorithmic}

\usepackage{amssymb}	
\usepackage{cases}
\definecolor{rouge}{rgb}{1,0,0}
\definecolor{bleu}{rgb}{0,0,1}
\definecolor{vert}{rgb}{0,1,0}

\title{\LARGE \bf Design of Software Rejuvenation for CPS Security Using Invariant Sets} 

\author{Raffaele Romagnoli$^\dagger$, Bruce H. Krogh$^\ddagger$ and Bruno Sinopoli$^\dagger$ \\
$^\dagger$Dept.of Electrical and Computer Engineering \\
$^\ddagger$Software Engineering Institute\\
Carnegie Mellon University \\
Pittsburgh, PA USA \\
\{rromagnoli$|$krogh$|$brunos\}@andrew.cmu.edu}

\begin{document}
\maketitle
\thispagestyle{empty}
\pagestyle{empty}

\begin{abstract}
Software rejuvenation has been proposed as a strategy to protect cyber-physical systems (CSPs) against unanticipated and undetectable cyber attacks. The basic idea is to refresh the system periodically with a secure and trusted copy of the online software so as to eliminate all effects of malicious modifications to the run-time code and data. Following each software refresh a safety controller assures the CPS is driven to a safe state before returning to the mission control mode when the CPS is again vulnerable attacks. This paper considers software rejuvenation design from a control-theoretic perspective. Invariant sets for the Lyapunov function for the safety controller are used to derive bounds on the time that the CPS can operate in mission control mode before the software must be refreshed and the maximum time the safety controller will require to bring the CPS to a safe operating state. With these results it can be guaranteed that the CPS will remain safe under cyber attacks against the run-time system and will be able to execute missions successfully if the attacks are not persistent. The general approach is illustrated using simulation of the nonlinear dynamics of a quadrotor system. The concluding section discusses directions for further research.   
\end{abstract}

\section{INTRODUCTION}
Cyber security has become a significant research and development issue for control of cyber-physical systems (CPSs), particularly for safety critical applications \cite{cardenas2008secure,slay2007lessons,case2016analysis,langner2011stuxnet}. The primary goal is to prevent physical damage to the system or environment due to malicious cyber attacks. Standard approaches to security rely on effective attack models. But attacks can be perpetrated in many different ways, and it is impossible to model all possible attacks.  Moreover, even the best detection mechanisms are impotent if an attacker is able to modify the run-time software.

Software rejuvenation, an established method for dealing with so-called software aging in traditional computing systems \cite{HKKF95}, has been proposed recently to deal with unmodeled and  undetectable cyber attacks on CPSs \cite{ASW17,arroyo2017fired}.   The idea is to periodically refresh the run-time system completely with a trusted, secure copy of the control software to thwart attacks that may have changed the on-line code.  Although the basic concept has been implemented for some demonstration systems, timing bounds for the software rejuvenation schedule need to be obtained  to assure that the software is refreshed before any damage has been done from malicious software and that the CPS will remain safe and viable after restarting the control software.  This paper derives these bounds using invariant sets for the case of controlling a CPS at an equilibrium state as a first step towards a sound theory of guaranteed security for general CPS using software rejuvenation.

\section{PREVIOUS WORK}
The concept of software rejuvenation was introduced by Huang et al. in 1995 \cite{HKKF95} to address the problem of so-called software aging; that is, failures that occur when a running program encounters a state that was not anticipated when the software was designed.  The basic idea is to restart the software intermittently at a "clean" state, either through complete system reboot or by returning to a recent checkpoint, with the hope that this will prolong the time until unanticipated states occur that might cause failures. Since the introduction of the concept, there has been considerable research into the development and performance of software rejuvenation strategies \cite{CNPR14}, and it has become a practical tool for enhancing the robustness of many computing systems \cite{ABL+12}.

Although software aging remains the primary motivation for implementing software rejuvenation strategies in computing systems, a few papers have proposed software rejuvenation to enhance system security \cite{AP04,LKGK14}. In contrast to software aging where mean-time to failure can be the basis for timing software refresh, the frequency of software refresh to defend against malicious attacks must be determined by the length of time a system can remain viable once its security has been violated. 

Arroyo et al. \cite{ASW17,arroyo2017fired} propose software rejuvenation as a strategy for CPS security and demonstrated the concept for a quadrotor controller and an automotive engine controller. These examples illustrate how refreshing software in a CPS impacts performance and introduces timing  constraints and safety considerations that aren't present in traditional computing systems.  

Abidi et al. \cite{Abdi2018} develop software rejuvenation for CPS further by introducing three concepts.  First, the \textit{hardware root of trust} is a secure onboard module that hosts the capabilities that must be available without compromise to implement software rejuvenation. Second, the \textit{secure execution interval} (SEI) is a period during which all external communication is disabled so that no cyber attacks can occur as the software is refreshed.  Third, the \textit{safety controller}, which executes immediately following a software refresh and during the SEI, drives the CPS to a known safe state before restoring communication and returning the system to mission control with vulnerability to attacks.  Abidi et al. use a simple, conservative reachability algorithm to determine the time that can be allowed before the next software reset and illustrate their approach for a simulated warehouse temperature controller and a bench-top 3-DOF helicopter.

In this paper, we adopt the overall approach to software rejuvenation from \cite{Abdi2018}, but modify the timing strategy based on analysis of the control problem.  We formulate the software rejuvenation problem for CPS applications where the safety controller is designed using the linearized dynamics to take the system to a neighborhood of a given equilibrium state.  Using invariant sets for the safety controller and system operating constraints, we derive bounds on the timing parameters for software rejuvenation strategy. We also introduce the concept of adding control constraints enforced by the hardware root of trust to increase the time allowed for mission control and software reset. We demonstrate our contributions using simulation of the nonlinear dynamics for a 6-DOF quadrotor.

\section{PROBLEM FORMULATION}\label{sec:formulation}
We use the following notation. For a vector $x \in \mathbb{R}^n$ and a symmetric $n \times n$ matrix $Q$,  $||x||_Q^2 = x^TQx$. For a given dynamic system, the set of states reachable at time $t>0$ from an initial state $x(0)=x_0$ under control policy $CP$ is denoted by $\mathcal{R}(x_0,t;CP)$, with the natural extension to a set of initial states.  When the control policy is clear from the context, the reachable set will be denoted as $\mathcal{R}(x_0,t;CP)$.  For a matrix $M$ with real eigenvalues, $\lambda_{min}(M)$ denotes the minimum eigenvalue of $M$.

\subsection{Software rejuvenation algorithm}
 Figure \ref{fig:TimingDiagram} illustrates the three operating modes for a CPS with software rejuvenation: \textit{mission control} (MC), where the CPS is executing its intended mission with network communication to other agents in the overall system, including possibly a supervisor; \textit{software refresh} (SR), when a secure copy of the operating software is reloaded to eliminate any possible corruption of the run-time code and data; and \textit{safety control} (SC), which is described below. The period for each of these operating modes is denoted $T_{MC}$, $T_{SR}$ and $T_{SC}$, respectively.  

\begin{figure}[h!]
\centering
\includegraphics[scale=.34,trim=125 155 12 150,clip]{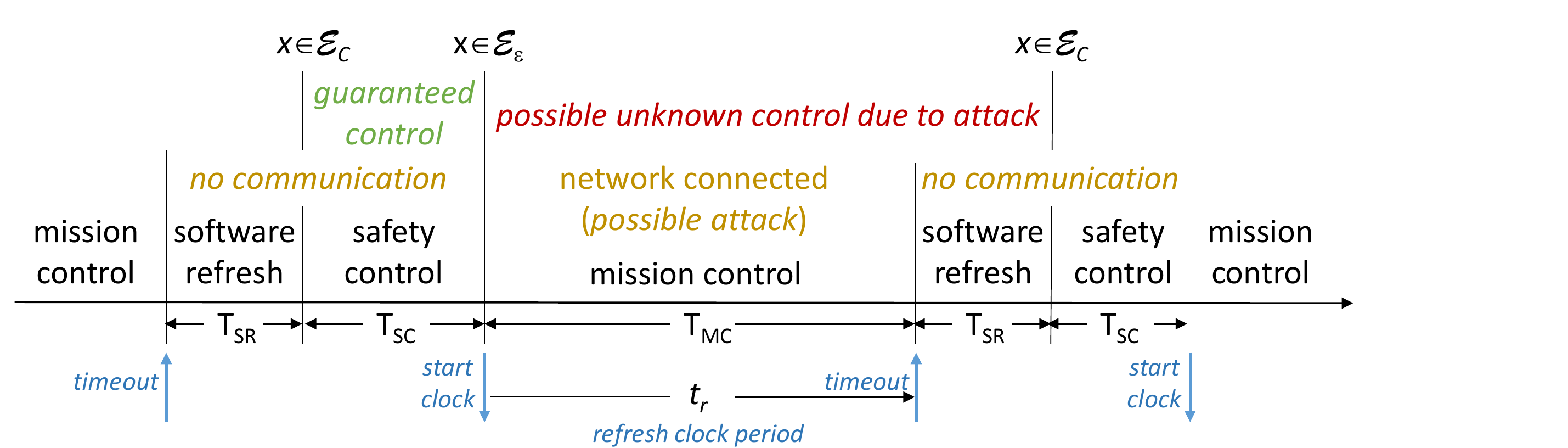} 
\caption{Software rejuvenation control modes.}\label{fig:TimingDiagram}
\end{figure}

Mission control is terminated by the timeout of the \textit{refresh clock}. To assure software refresh and safety control executed without interference from potential attackers, external communication is shutoff and the CPS operates autonomously during the SEI ($T_{SR}+T_{SC}$). After execution of the safety control, the refresh clock is restarted  with the refresh clock period $t_r$ (which determines $T_{MC}$, communication is re-established and th CPS returns to mission control.  Attacks can occur when there is communication. If there is an attack, the control actions may not be known to the CPS controller.  This uncertainty about the control actions (and the state of the CPS) remains through the software refresh. 

The pseudo code in Alg. \ref{alg:SoftwareRejuvenation} illustrates the implementation of the overall software rejuvenation strategy. The protected secure hardware hosts: system communication; trusted software image; refresh clock; and control limits that can be imposed during MC and software SR to constrain what an adversary could do during this period of uncertain control. Since the system is operating securely during SC, these limits can be removed for SC so that it can use the full control capability.

\begin{algorithm}
\caption{Software Rejuvenation Algorithm}\label{alg:SoftwareRejuvenation}
\begin{algorithmic}[1]
\STATE Initialize: 
\STATE \hspace{0.2cm} COMMUNICATION OFF;  (protected)
\STATE \hspace{0.2cm} LOAD SOFTWARE; (protected)
\STATE Step 1:
\STATE \hspace{0.2cm} OFF $\rightarrow$ CTRL\_LIMS; (protected)
\STATE \hspace{0.2cm} WHILE $x \not \in \mathcal{E}_{\epsilon}$
\STATE \hspace{0.6cm} SAFETY\_CONTROL;
\STATE \hspace{0.2cm} END
\STATE \hspace{0.2cm} ON $\rightarrow$ CTRL\_LIMS; (protected)
\STATE \hspace{0.2cm} $t_r$ $\rightarrow$ REFRESH\_CLK; (protected)
\STATE \hspace{0.2cm} COMMUNICATION ON; (protected)
\STATE Step 2:
\STATE \hspace{0.2cm} UNTIL  timeout $\leftarrow$ REFRESH\_CLK (protected)
\STATE \hspace{0.6cm} MISSION\_CONTROL;
\STATE \hspace{0.2cm} END
\STATE Step 3:
\STATE \hspace{0.2cm} COMMUNICATION OFF; (protected)
\STATE \hspace{0.2cm} REFRESH SOFTWARE; (protected)
\STATE \hspace*{0.2cm} GO TO Step 1;
\end{algorithmic}
\end{algorithm}

\subsection{Guaranteeing safety}

To guarantee unknown control actions do not drive the CPS into unsafe or undesired operating states, certain guarantees need to be enforced about the possible states reached during the period of \textit{uncertain control} (UC), $T_{UC}=T_{MC}+T_{SR}$. This is the period when an attacker may have taken over control of the system. The conditions that need to be satisfied are represented by two sets of states, the \textit{safe set} and the \textit{inner safe set}, denoted $\mathcal{E}_{\mathcal{C}}$ and $\mathcal{E}_{\epsilon}$, respectively. Figure \ref{fig:SafeSets} illustrates the role of these set. The safe set represents general operating constraints to assure system safety and controllability. The goal of the safety controller is to assure the system remains in $\mathcal{E}_{\mathcal{C}}$ by always returning the state to the inner-safe set before handing control back to the mission controller.  During the period $T_{UC}$, the state may leave $\mathcal{E}_{\epsilon}$, but safety control must be executed before the state leaves $\mathcal{E}_{\mathcal{C}}$.  
\begin{figure}[h!]
\centering
\includegraphics[scale=1]{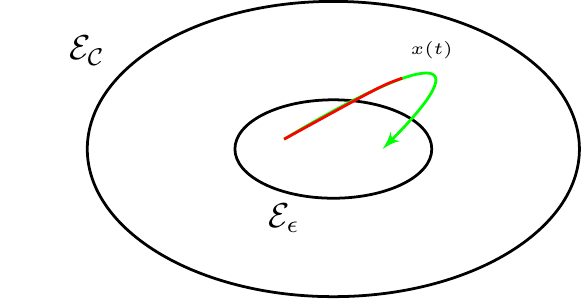} 
\caption{State constraints for software rejuvenation. Red: uncertain control; green: safety control.}
\label{fig:SafeSets}
\end{figure}




General conditions that guarantee the software rejuvenation algorithm will keep the system in $\mathcal{E}_{\mathcal{C}}$ are summarized by the following proposition \cite{Aub09}.


\begin{prop} \label{prop1}
Given a CPS with time-invariant dynamics and a set of safe states $\mathcal{E}_{\mathcal{C}}$, a set of inner safe states $\mathcal{E}_{\epsilon} \subset \mathcal{E}_{\mathcal{C}}$, software refresh time $T_{SR}$ and a safety controller $SC$, if
\begin{enumerate}[i.]
    \item $x(0) \in \mathcal{E}_{\mathcal{C}}$;
    \item $\exists \ \overline{T}_{SC} > 0 \ \ni \ \forall \ x \in \mathcal{E}_{\mathcal{C}}, \ \mathcal{R}(x,\overline{T}_{SC};SC) \subseteq \mathcal{E}_{\epsilon}$; and
    \item $\exists \ T_{UC} > T_{SR} \ni \ \forall \ 0 \leq t \leq T_{UC}, \ \mathcal{R}(\mathcal{E}_{\epsilon},t;UC) \subseteq \mathcal{E}_{\mathcal{C}}$;
\end{enumerate}
then for the CPS controlled by the software rejuvenation algorithm with $t_r = T_{UC}-T_{SR},  \ x(t) \in \mathcal{E}_{\mathcal{C}} \ \forall \ t \geq 0$.
\vspace{0.1in}
\end{prop}

Proposition \ref{prop1} provides conditions for designing safe software rejuvenation strategies. Condition (i) requires that the system starts at a safe state.  Condition (ii) indicates the safety controller must be able to drive the system from all safe states to the inner safe set in bounded time. The amount of time SC executes depends on how long it takes to drive the state to $\mathcal{E}_{\epsilon}$ (Alg. \ref{alg:SoftwareRejuvenation} line 6). Condition (iii) indicates that $T_{UC}$ must be sufficient for software refresh and small enough that no admissible control can take the system state from the inner safe set to an unsafe state before safety control is initiated.  

The software refresh time $T_{SR}$ is determined by the system hardware and the size of the software image. $\mathcal{E}_{\mathcal{C}}$ should be made as large as possible to maximize the time allowed for mission control.  For a given SC, design of the software rejuvenation algorithm depends on the size of $\mathcal{E}_{\epsilon}$ and the control constraints imposed during $T_{MC}$ and $T_{SR}$.   $T_{UC}$ can be increased by reducing the size of $\mathcal{E}_{\epsilon}$ and by making the control limits more conservative (Alg. \ref{alg:SoftwareRejuvenation} line 9). But a smaller $\mathcal{E}_{\epsilon}$ can lead to longer times for SC and tighter control limits reduce the control available for MC, so there are clear design trade offs to be considered.  

The remainder of the paper focuses on methods for satisfying the conditions in Prop. \ref{prop1} for the case of controlling a CPS in the neighborhood of an equilibrium point.  We assume a supervisor and mission controller always keep the CPS in the neighborhood of some equilibrium point, which may be changing as a mission is executed. What's important is that the current equilibrium $x_{e}$ is known and the state of the system is always with $\mathcal{E}_{\mathcal{C}}$ of the equilibrium state. Section \ref{sec:safety_controller} shows how to find a maximal ellipsoidal safe set $\mathcal{E}_{\mathcal{C}}$ for set of operating constraints defined by affine constraints in the state space $\mathcal{E}_{\mathcal{C}}$, and how to define a corresponding class of inner safe states $\mathcal{E}_{\epsilon}$ such that for any stabilizing safety controller $\overline{T}_{SC}$ can be found to satisfy condition (ii). Section \ref{sec:reset time} shows how to use reachabilty computations to find a $T_{UC}$ that satisfies condition (iii).

\section{SAFETY CONTROLLER}\label{sec:safety_controller}
The dynamics for the CPS linearized at the equilibrium state $x_e$ are given by the LTI state equations
\begin{eqnarray}
\dot{x}(t)&=& A x(t)+Bu(t)\quad t \in \mathbb{R}^+, \label{sys}
\end{eqnarray}
where $x \in \mathbb{R}^n$ and $u \in \mathbb{R}^m$. Assuming the state is  measurable and the system is controllable, the safety controller is designed as a state feedback control law of the form
\begin{equation}
u(t)=-K x(t) \label{u}
\end{equation}
where $K$ is a gain matrix chosen to make the closed-loop system matrix $A_{SC}\triangleq(A-BK)$ Hurwitz. Thus, under safety control the system dynamics are stable and given by 
\begin{equation}
\dot{x}(t)=A_{SC}x(t). \label{cloop}
\end{equation}

We assume the safe states can be represented by a convex polyhedral region $\mathcal{C} \subset \mathbb{R}^n$ of the state space given by
\begin{equation}
\mathcal{C}=\left\lbrace x\; |\;  \xi_j^T x\leq 1,\; j=1,...,n_c \right\rbrace \label{C1}
\end{equation}
where $n_c$ represents the number of constraints representing the operating constraints and the region about the equilibrium where \eqref{sys} and \eqref{cloop} hold with the control \eqref{u} adhering to the control limits $u \in \mathcal{U}$.

Let us define the \textit{safe set} $\mathcal{E}_{\mathcal{C}}$ as the largest positively invariant ellipsoid for \eqref{cloop} contained in $\mathcal{C}$.  This largest positively invariant ellipsoid $\mathcal{E}_{\mathcal{C}}$ can be found solving the following maximization problem \cite{boyd1994linear}:
\begin{equation}\label{safeinv}
\left\lbrace\begin{array}{l}
\mathrm{max\ log\ det} Q\\
s.t.\ Q A^{T}_{SC} + A_{SC}Q\leq 0\\
\qquad\!\! \xi_j^T Q \xi_j\leq 1,\; j=1,...,n_c\\
\qquad\!\! Q>0
\end{array} \right.
\end{equation}
Given a solution $Q$ to \eqref{safeinv}, define $P=Q^{-1}$.  Assuming no control saturation, $V(x)=||x||_P^2$ is a Lyapunov function for \eqref{cloop} with
\begin{equation}
\dot V(x) = A_{SC}^TP+ P A_{SC} < 0. 
\end{equation}
The positively safe invariant set is given by
\begin{eqnarray}\label{E_C}
 \mathcal{E}_{\mathcal{C}}=\left\lbrace x\;|\; ||x||_P^2 \leq 1 \right\rbrace.
\end{eqnarray} 
Note that $\mathcal{E}_{\mathcal{C}}$ is also defined by the set of states satisfying $V(x) \leq 1$.

Given $0 < \epsilon < 1$, we define the invariant ellipsoid
\begin{equation}\label{E_epsilon}
\mathcal{E}_\epsilon \triangleq \left\lbrace x\;|\; V(x) = ||x||_P^2 \leq \epsilon \ \right\rbrace = \epsilon\mathcal{E}_{\mathcal{C}}.
\end{equation}
We want to show that there exists a $\overline{T}_{SC}$ for which condition (ii) of Prop. \ref{prop1} is satisfied.

For any state trajectory $x(t)$ of \eqref{cloop} starting at a given initial condition $x(0)=x_0$, it can be shown that \cite{kalman1960control}
\begin{equation}
    V(x(t))\leq e^{-\gamma(t)}V(x(0)),\; \forall \ t\geq 0, \label{Lyap_bound}
\end{equation}
with
\begin{equation}
 \gamma= \min_{x} \frac{-\dot{V(x)}}{V(x)} = \lambda_{min}(WP^{-1}), \label{gamma}  
\end{equation}
where $W\triangleq -A_{SC}^TP-PA_{SC}$.  Therefore, \eqref{Lyap_bound} implies a bound on the maximum time it can take for the safety controller to drive the state of \eqref{cloop} to $x(t) \in \mathcal{E}_{\epsilon}$ (for which $V(x(t)) \leq \epsilon$) from any initial state $x_0 \in \mathcal{E}_{\mathcal{C}}$ (for which $V(x_0) \leq 1$) is given by
\begin{equation}
    \overline{T}_{SC} = \frac{-\mathrm{ln}(\epsilon)}{\gamma}. \label{T_max}
\end{equation}
We note that in general \eqref{T_max} is a very conservative bound on the worst-case value of $T_{SC}$.

\section{UNCERTAIN CONTROL PERIOD}\label{sec:reset time}
This section presents a procedure to find a value of the uncertain control period $T_{UC}$ that satisfies condition (iii) of Prop. \ref{prop1}. Two conditions must be satisfied:
\begin{equation}
T_{UC}>T_{SR} \label{TUCgtTSR}
\end{equation}
and
\begin{equation}
    \forall \ 0 \leq t \leq T_{UC}, \ \mathcal{R}(\mathcal{E}_{\epsilon},t;UC) \subseteq \mathcal{E}_{\mathcal{C}}, \label{TUCReach}
\end{equation}
where for \eqref{sys}
\begin{eqnarray*}
    &\hspace{-1.2cm} \mathcal{R}(\mathcal{E}_{\epsilon},t;UC)=  \{x| e^{At}x_0 + \int_0^t e^{A(t-\tau)}Bu(\tau)d\tau,  \\
 &\hspace{3.0cm}x_0 \in \mathcal{E}_{\epsilon}, \ u(\tau) \in \mathcal{U}, \ 0 \leq \tau \leq t \}.  
\end{eqnarray*}

We first consider condition \eqref{TUCReach}. Following the procedure from \cite{hwang2005polytopic}, the maximum principle can be applied to overapproximate the reach set at each point in time with a set of supporting hyperplanes. We apply this procedure using a polyhedral representation of $\mathcal{U}$ and a polytope $\mathcal{P}$ that contains $\mathcal{E}_{\epsilon}$ as the approximation of $\mathcal{R}(\mathcal{E}_{\epsilon},0;UC)$ and continue computing  $\mathcal{R}(\mathcal{P},t;UC)$ which contains $\mathcal{R}(\mathcal{E}_{\epsilon},t;UC)$ for increasing values of $t$ until the overapproximation is not contained in $\mathcal{E}_{\mathcal{C}}$. The largest value of $t$ for which $\mathcal{R}(\mathcal{P},t;UC) \subset \mathcal{E}_{\mathcal{C}}$ provides a value for $T_{UC}$.

Entering into the details of the method of the computation of the reach set described in \cite{hwang2005polytopic}, let us start from the polytopic approximation of $\mathcal{E}_\epsilon$ \begin{eqnarray}
    \mathcal{P}=\left\lbrace x\;|\; \alpha_i^T x \leq 1,\;  i=1,...,n_c \right\rbrace \supset \mathcal{E}_\epsilon, 
\end{eqnarray}
where $n_c$ represents the number of constraints.
The supporting hyperplanes of the reach set $\mathcal{R}(\mathcal{P},t;UC)$ are
\begin{equation}
    v^+_i(x,t)=\alpha_i(t)x
\end{equation}
where $\alpha_i(t)$ are the normal directions computed as
\begin{equation}
    \alpha_i(t)=e^{-A^T\tau}\alpha_i. \label{normal_dir}
\end{equation}
Expression \eqref{normal_dir} is obtained applying the Pontryagin maximum principle of optimal control theory \cite{varaiya2000reach}.
Then, defining $\left\lbrace u^1,u^2,...,u^{n_v} \right\rbrace$ the $n_v$ vertexes that describise $\mathcal{U}$,   an over approximation of the reachable set $\mathcal{R}(\mathcal{P},t;UC)$ is 
\begin{eqnarray}
    &\hspace{-1cm}\mathcal{R}^+(\mathcal{P},t;UC)= \bigcap_{i=1}^{n_c}\left\lbrace x|v^+_i(x, t) \leq \right. \nonumber \\
    &\hspace{2.5cm}\left. \int_0^{ t} \max_j\;\langle\alpha_i(\tau),Bu_j \rangle d\tau + \right.\nonumber \\
    &\hspace{3.8cm}\left. \max_{x(0)\in \mathcal{P}} v^+_i(x(0),0) \right\rbrace. \label{V+1} 
\end{eqnarray}
Those bounds are generated considering the supporting hyperplanes $v_i(x,t)$ as a solution of the Hamilton-Jacobi-Isaacs partial differential equation  \cite{kurzhanski2001dynamic}.

Since $\mathcal{R}^+(\mathcal{P},t;UC)$ can be also represented in terms of $n_r$ vertexes defined as $x^1,x^2,...,x^{n_r}$, 

\begin{equation}
    \mathcal{R}^+(\mathcal{P}, t; UC) \subseteq \mathcal{E}_\mathcal{C} \iff x^{i^T}P x^i \leq 1, \label{cond1}
\end{equation}
for $i=1,...,n_r$.
The largest value of $t$ such that $\mathcal{R}^+(\mathcal{P},t;UC) \subset \mathcal{E}_{\mathcal{C}}$, is the largest value that satisfies \eqref{cond1}.

If the computed value of $T_{UC}\leq T_{SR}$, condition \eqref{TUCgtTSR} has not been satisfied and two strategies can be used to try to increase the value $T_{UC}$: (a) $\epsilon$ can be reduced; and (b) the constraints on the control can be reduced.  Strategy (a) increases the distance between the boundary of  $\mathcal{E}_\epsilon$ and the boundary of $\mathcal{E}_\mathcal{C}$, meaning that it will take longer to drive the system out of the safe set. This results in a smaller inner safe set, however, which will increase the bound on the worst-case execution time for the safety controller.  Strategy (b) will decrease the rate at which an attacker can drive the system out of set of safe states. But the same control limits will apply for the mission controller, so increased security will be achieved at the expense of reducing system performance.  We illustrate these strategies to find a feasible value for $T_{UC}$ in the following section.

\section{EXAMPLE}\label{sec:example}
The proposed method is tested on a quadrotor system, which represents a challenge for secure control scheme based on software rejuvenation \cite{arroyo2017fired} \cite{Abdi2018}. The considered quadrotor is the "Generic 10" Quad + geometry" used  in the PX4 flight controller platform with full thrust of 4 N, and full torque of 0.05 Nm for each motor \cite{px4}. The quadrotor is described by a nonlinear mathematical model with 12 state variables, 3 for the position in the inertial frame of reference, 3 for the orientation of the airframe (roll, pitch, yaw), 3 for the positional velocities, and 3 for the angular velocities \cite{beard2008quadrotor}. The control variables are the force normal to the air frame and the torques about the three angles of orientation.  A matrix multiplication (the "mixer") maps these control variables to the motor torques. 

In general, the controller is designed considering a linearized model of the system around an equilibrium point \cite{bolandi2013attitude} \cite{bouabdallah2004pid} \cite{araar2014full}. Without loss of generality that equilibrium point can be considered equal to origin of the state space.  It is important to note that the controller performance depends on how close the system is operating to the equilibrium point. Moving the system away from that equilibrium point the controller could not work in a proper way. In this scenario, it can be easy for an attacker to pursue physical damage to the system.

The controller used to stabilize the linearized system is an LQR controller with integral action \cite{araar2014full}. The goal is to keep the quadrotor hovering at the equilibrium point where the system has been linearized.

For implementing the proposed software rejuvenation scheme, the safe invariant ellipsoid  $\mathcal{E}_\mathcal{C}$ is computed solving the maximization problem \eqref{safeinv} based on the set $\mathcal{C}$ that in this case is represented by a set of constraints on each of the 12 state variables. For the position, the vertical direction is bounded between $\pm$ 5 m and the other two directions between $\pm$ 2 m. The vertical velocity is bounded by $\pm$ 5 m/s, and the horizontal velocity limits are $\pm$ 2 m/s. The bounds on the angles of roll pitch and yaw are fundamental for good performance of the linear controller. In this case all the angles are bounded between $\pm\pi/4$. The angular velocities are between $\pm$ 5 rad/s. 

Following the procedure described in Section V, a feasible value of $T_{UC}$ is $0.37$ s for $\epsilon=0.01$, with $T_{SR}=0.1$s. The value of $T_{UC}$ has been computed considering the following reduced bounds on the torques, $\pm$ 0.0033 Nm for roll and pitch, and $\pm$ 0.0005 Nm for the yaw. It is important to remark that, those bounds are not imposed during the SC mode, where only the less restrictive constraints are considered, namely 0.05 Nm \cite{px4}.

Figure \ref{fig:invreach} shows shows projections of the ellipsoids $\mathcal{E}_{\mathcal{C}}$ (light blue) and $\mathcal{E}_{\epsilon}$ (light green) on the space of the positions (left) and angles (right). The figure also shows the projections of the reach set $\mathcal{R}^+(\mathcal{P}_\epsilon,T_{UC})$ (blue polytope) on the sub-spaces. From the figure it is possible to observe that all vertexes are completely contained in the projections of the invariant ellipsoid $\mathcal{E}_{\mathcal{C}}$.

\begin{figure}[h!]
	\centering
	\begin{subfigure}{} 
		\includegraphics[scale=0.65]{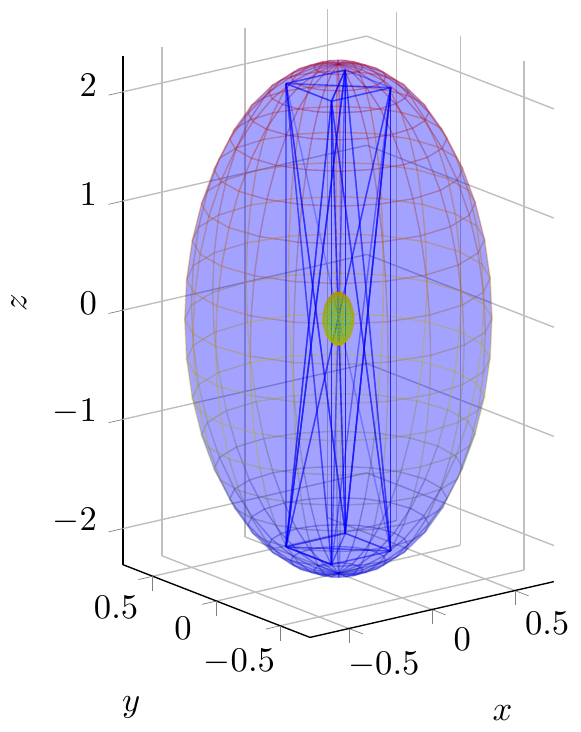}
	\end{subfigure}
	\begin{subfigure}{} 
		\includegraphics[scale=0.65]{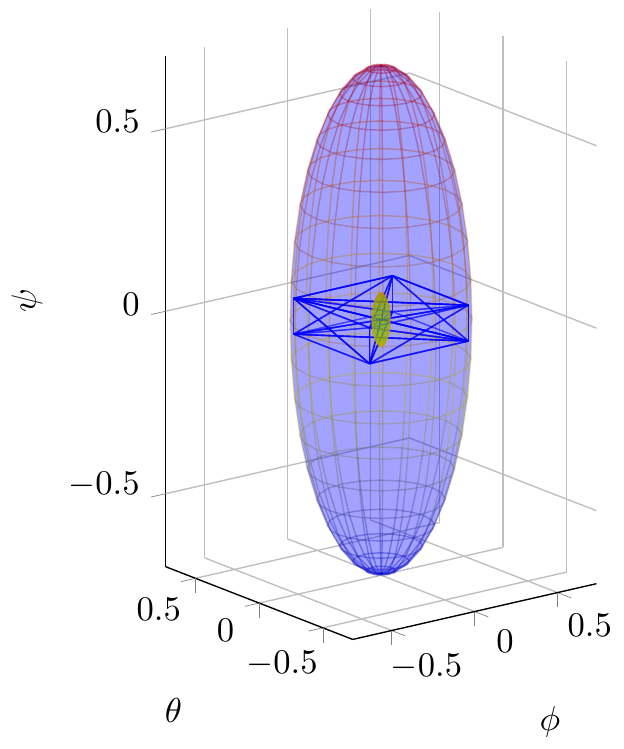}
	\end{subfigure}
	\caption{Projections of $\mathcal{E}_{\mathcal{C}}$(light blue), $\mathcal{E}_{\epsilon}$ (light green) and $\mathcal{R}^+$ computed for $\epsilon=0.01$ and $T_{UC}=0.37s$ (blue line) on the position (left), and angles(right).} 
	\label{fig:invreach}
\end{figure}

Using the values $T_{UC}=0.37$ s and $\epsilon=0.01$ the algorithm of software rejuvenation has been tested on the nonlinear model of the quadrotor controlled by the LQR controller developed using the linearized model of the system. Several tests have been carried out simulating attacks that provide arbitrary control inputs. In all the cases the reset clock $t_r$ and safety controller keep the state of the system from hitting the boundary of $\mathcal{E}_{\mathcal{C}}$. 

Figure \ref{fig:sim} shows the position behaviour of the quadrotor respect the invariant ellipsoids, and  Fig. \ref{fig:sim1} shows the zoom of the behaviour respect the several modes of the controller. Figure \ref{fig:timeline} shows the time-line for the scheduling of the software rejuvenation modes in a scenario with two attacks. In the first case the attacker turns off all the propellers, and in the second the case the attacker tries to drive the system to a different  equilibrium point.

\begin{figure}[h!]
    \centering
    \includegraphics[scale=0.8]{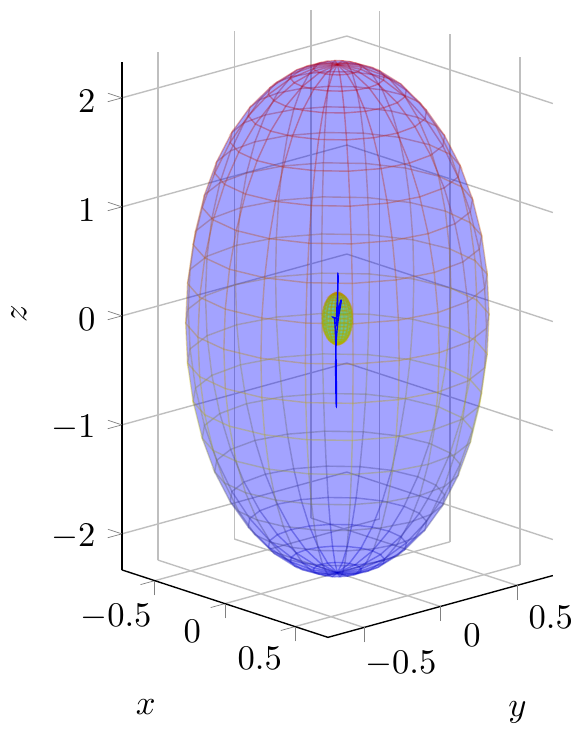}
    \caption{Quadrotor position behaviour respect to $\mathcal{E}_{\mathcal{C}}$(light blue) and $\mathcal{E}_{\epsilon}$ (light green).} 
    \label{fig:sim}
\end{figure}

\begin{figure}[h!]
    \centering
    \includegraphics[scale=0.6]{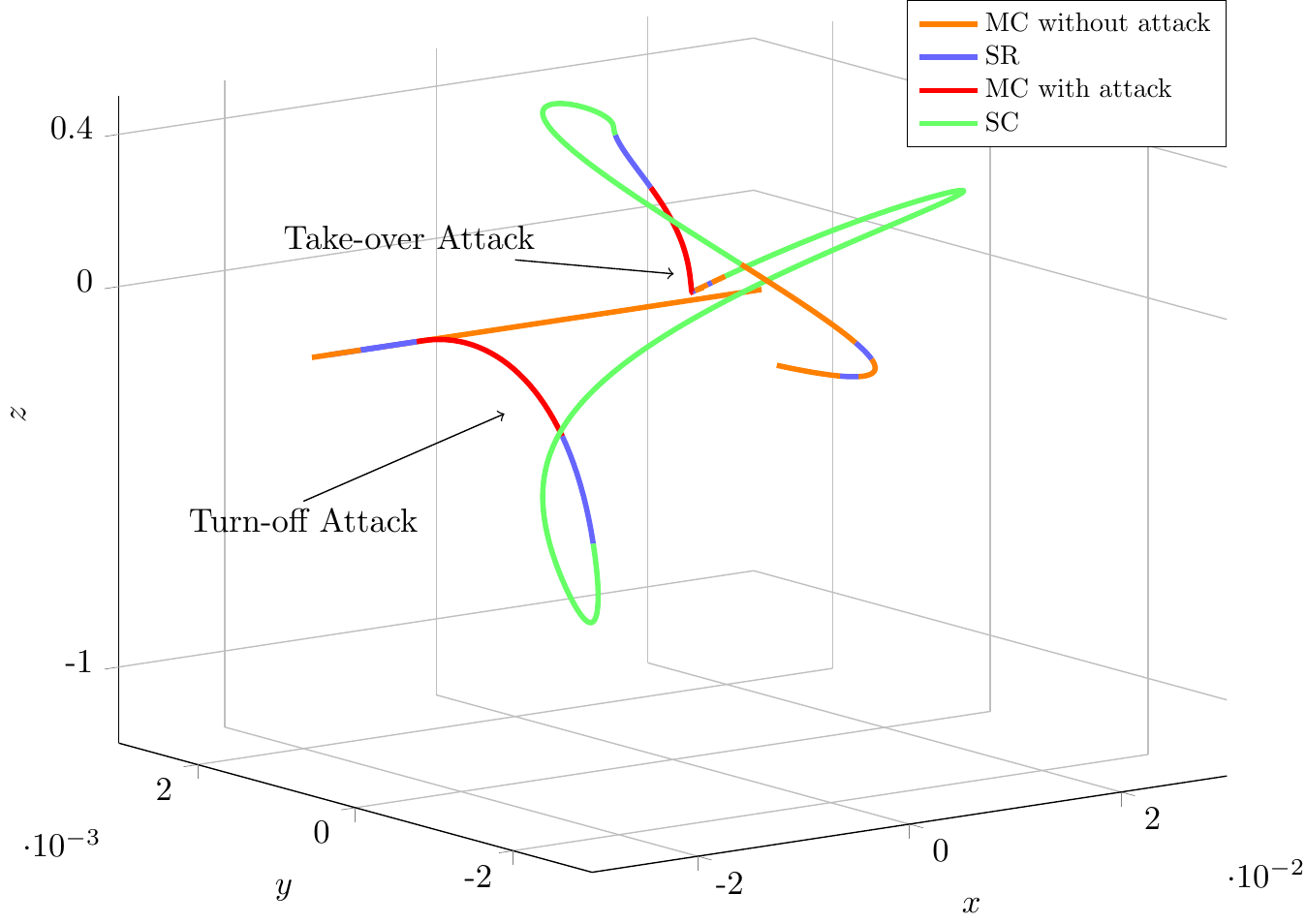}
    \caption{Zoom of the position behaviour during the several modes: MC without attack (orange), MC with attack (red), SR (blue), and SC (green).} 
    \label{fig:sim1}
\end{figure}


\begin{figure}[h!]
\centering
\includegraphics[scale=0.6]{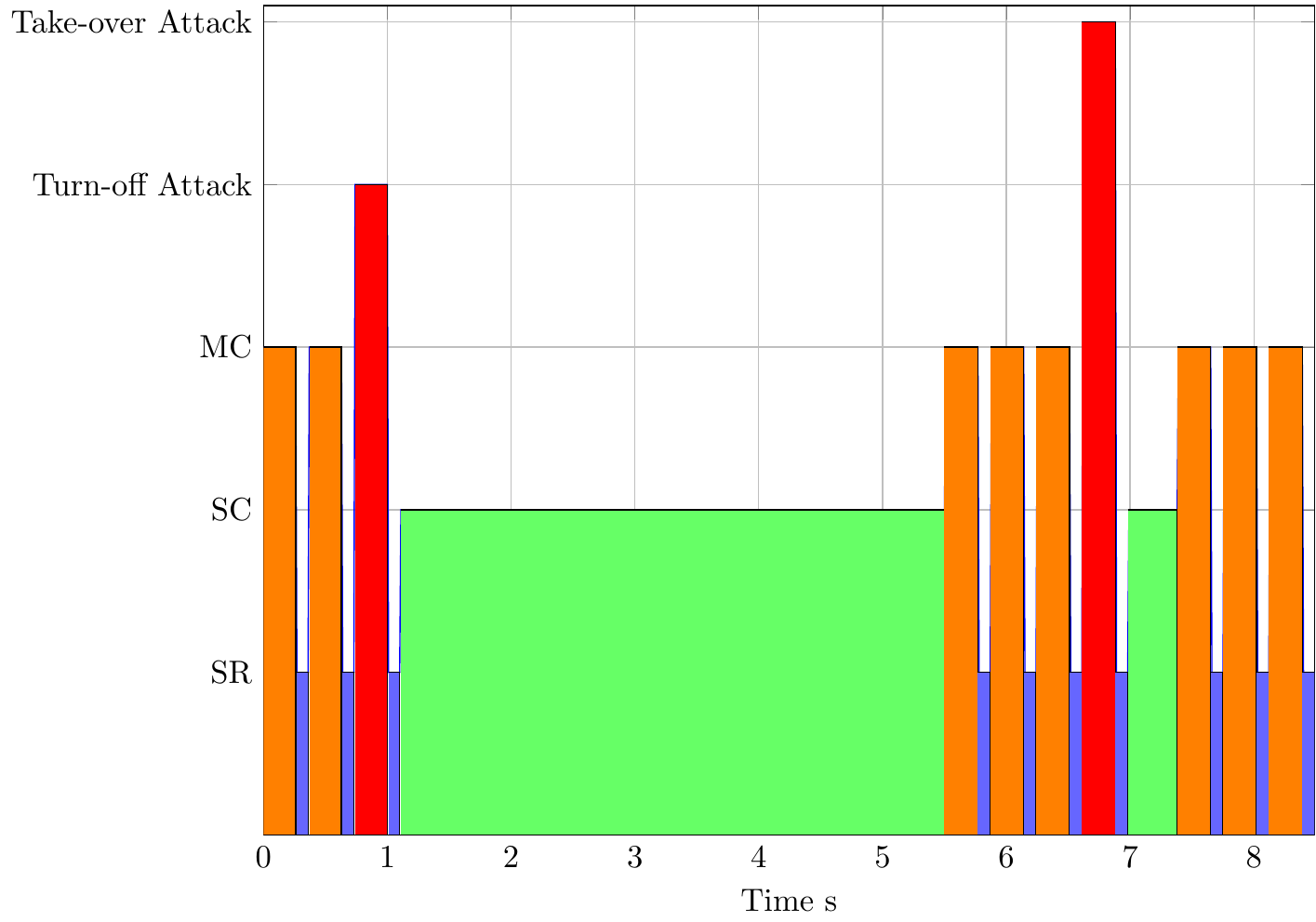} 
\caption{Time scheduling diagram respect the several conditions, $SR$ Software Refresh (blue), $SC$ Safety Control (green), $MC$ mission control (orange), and mission control under attacks: Turn-off Attack and Take-over Attack (red).}
\label{fig:timeline}
\end{figure}

The results show the effectiveness the adopted solution using $T_{UC}=0.37$ s and $\epsilon=0.01$ Fig.\ref{fig:sim}. In Fig.\ref{fig:timeline}the controller is operating in MC (orange), SR (blue), SC (green), and MC under attack (red). It is important to note that, in the case of normal operation during the MC mode, the quadrotor remains in $\mathcal{E}_{\epsilon}$ after SR, so SC is not activated. If the quadrotor is out of $\mathcal{E}_\epsilon$ after SR, SC is activated and this happens in particular after the two attacks. Note that the SC activation does not require the attack detection; the activation strategy is based only on verifying the state condition $x^TPx\leq1$.

\section{DISCUSSION}\label{sec:discussion}
This paper presents conditions for safety controller design and timing parameters to support the implementation of software rejuvenation to provide CPS security against cyber attacks on the run-time code and data. The approach for the case of linear dynamics at an equilibrium state is developed the general results are illustrated using simulation of a quadrotor, demonstrating how safety can guaranteed without having to detect the cyber attacks.  The results in this paper provide the foundation for implementing tracking control, where a supervisory generates the reference signal to the mission control algorithm in a manner that provides a sequence of equilibrium states for safety control. In contrast to the approach developed in \cite{Abdi2018}, the safety controller executes only when it is needed, with an upper bound on the time required to return the system to a safe state, and the reachability set computations to determine the software refresh time are performed offline rather than at run time.  

There are several directions for future research. The safety control time bound and the reset time can be improved using less conservative reachability algorithms. We are currently extending the results to incorporate state estimation, modeling uncertainties and disturbances. We are also developing experimental implementations and investigating methods for managing the equilibrium information and control limits on protected hardware components.   
\section*{Acknowledgments}
\begin{spacing}{0.8}
{\footnotesize 
\noindent Copyright 2018 IEEE. All Rights Reserved. This material is based upon work funded and supported by the Department of Defense under Contract No. FA8702-15-D-0002 with Carnegie Mellon University for the operation of the Software Engineering Institute, a federally funded research and development center. This material has been approved for public release and unlimited distribution.  Please see Copyright notice for non-US Government use and distribution.
Carnegie Mellon® is registered in the U.S. Patent and Trademark Office by Carnegie Mellon University.
}
\end{spacing}
 

\bibliographystyle{unsrt}
\bibliography{ifacconf7.bib} 
\end{document}